\documentclass{article}

\usepackage{amsthm, amssymb, amsmath}
\usepackage{verbatim}
\usepackage{graphics}      
\usepackage{graphicx}      
\usepackage{setspace}

\newtheorem{thm}{Theorem}

\newtheorem{cor}[thm]{Corollary}
\newtheorem{lem}[thm]{Lemma}

\theoremstyle{plain}

\newcommand{\gns}{\textrm{GNS}}
\newcommand{\pr}{\textrm{Pr}}

\newcommand{\R}{\mathbb{R}}
\newcommand{\E}{\textrm{E}}

\newcommand{\sgn}{\textrm{sgn}}

\title{The Gaussian Surface Area and Noise Sensitivity of Degree-$d$ Polynomials}
\author{Daniel M. Kane}

\begin{document}

\maketitle

\section{Introduction}

We provide asymptotically sharp bounds for the Gaussian surface area and the Gaussian noise sensitivity of polynomial threshold functions.  In particular we show that if $f$ is a degree-$d$ polynomial threshold function, then its Gaussian sensitivity at noise rate $\epsilon$ is less than some quantity asymptotic to $\frac{d\sqrt{2\epsilon}}{\pi}$ and the Gaussian surface area is at most $\frac{d}{\sqrt{2\pi}}$.  Furthermore these bounds are asymptotically tight as $\epsilon\rightarrow 0$ and $f$ the threshold function of a product of $d$ distinct homogeneous linear functions.

The noise sensitivity and surface area are both of fundamental interest and useful in the analysis of agnostic learning algorithms (see \cite{kn:learning}). In particular our results imply that the class of degree-$d$ polynomial threshold functions is agnostically learnable under the $n$-dimensional Gaussian distribution in time $n^{O(d^2/\epsilon^4)}.$

A number of other authors have attempted to prove bounds along these lines. \cite{kn:relate} proves a bound on noise sensitivity in terms of surface area and we relate our bounds essentially by also proving the other direction of this inequality for boolean functions with smooth interface that switch signs a bounded number of times on any line through the origin.  Our bounds are obtained via a simple computation in the case of $d=1$.  A bound of $\tilde{O}(\epsilon^{1/(2d)})$ noise sensitivity was recently proved by \cite{kn:DRST} and independently by \cite{kn:HKM} for multilinear polynomials.

There is also interest in related questions for points picked uniformly from vertices of the hypercube rather than with the Gaussian distribution.   It is conjectured in \cite{kn:GL} that the corresponding noise sensitivity in this case is also always $O(d\sqrt{\epsilon}).$  The $d=1$ case of this conjecture was proved by \cite{kn:P}, improving upon a bound of $O(\epsilon^{1/4})$ of \cite{kn:BKS}.  It is noted in \cite{kn:DRST} that such a result would imply a similar bound for the Gaussian case.  Hence our results can be thought of as a first step toward proving this conjecture.

\subsection{Basic Definitions}

Given a function $f:\R^n\rightarrow \{-1,1\}$ we define the Gaussian noise sensitivity at noise rate $\epsilon$ as
$$
\gns_\epsilon(f) := \pr(f(X) \neq f(Z))
$$
where $X$ is an $n$-dimensional Gaussian random variable, and $Z=(1-\epsilon)X+\sqrt{2\epsilon-\epsilon^2}Y$ for $Y$ an independent $n$-dimensional Gaussian.

This is closely related to the Gaussian surface area of $f^{-1}(1)$.  In particular we define the Gaussian surface area of a set $A$ to be
$$
\Gamma(A) := \liminf_{\delta\rightarrow 0} \frac{\textrm{GaussianVolume}(A_\delta\backslash A)}{\delta}.
$$
Where the Gaussian volume of a region $R$ is $\pr(X\in R)$ for $X$ a Gaussian random variable, and where $A_\delta$ is the set of points $x$ so that $d(x,A)\leq \delta$ (under the Euclidean metric).  We note that if $A$ is an open region whose boundary is smooth away from codimension 2, that its Gaussian surface area is equal to
$$
\int_{\partial A} \phi(x) d\sigma.
$$
Where $\phi(x)$ is the Gaussian density, and $d\sigma$ is the surface measure on $\partial A$.  Furthermore if $A$ is such a region, then its Gaussian surface area is seen to be equal to
$$
\lim_{\delta\rightarrow 0} \frac{\textrm{GaussianVolume}((\partial A)_\delta)}{2\delta}.
$$
For $f$ a boolean function, we define
$$
\Gamma(f) := \Gamma(f^{-1}(1)).
$$

The concepts of noise sensitivity and surface area are related to each other by noting that the noise sensitivity is roughly the probability that $X$ is close enough to the boundary that wiggling it will push it over the boundary.

\subsection{Statement of Results}

We focus on proving two main results.  We define $f$ to be a degree $d$ polynomial threshold function if $f(x)=\sgn(p(x))$ for some degree $d$ polynomial $p$.  We prove the following Theorems about such functions:

\begin{thm}\label{nsthm}
If $f$ is a degree $d$ polynomial threshold function, then
$$
\gns_\epsilon(f) \leq \frac{d\arcsin(\sqrt{2\epsilon-\epsilon^2})}{\pi} \sim \frac{d\sqrt{2\epsilon}}{\pi} = O(d\sqrt{\epsilon}).
$$
Furthermore this bound is asymptotically tight as $\epsilon\rightarrow 0$ for the threshold function of any product of distinct linear functions.
\end{thm}

\begin{thm}\label{surfaceareatheorem}
If $f$ is a degree-$d$ polynomial threshold function then $\Gamma(f) \leq \frac{d}{\sqrt{2\pi}}.$
\end{thm}

Section 2 will be devoted to the proof of Theorem \ref{nsthm}, Section 3 to the proof of Theorem \ref{surfaceareatheorem}, and Section 4 will provide some closing notes.

\section{Proof of the Noise Sensitivity Bound}

\begin{proof}[Proof of Theorem \ref{nsthm}]
We begin by letting $\theta = \arcsin(\sqrt{2\epsilon-\epsilon^2})$.  We need to bound
\begin{equation}\label{sensitivity}
p:=\gns_\epsilon(f) = \pr(f(X)\neq f(\cos(\theta)X+\sin(\theta)Y)).
\end{equation}
We note that the value of $p$ given in Equation \ref{sensitivity} remains the same if $X$ and $Y$ are replaced by any $X'$ and $Y'$ that are i.i.d. Gaussian distributions.  In particular we define
$$
X_\phi = \cos(\phi)X + \sin(\phi)Y.
$$
Note that $X_{\phi}$ and $X_{\phi+\pi/2}$ are i.i.d. Gaussians.  Using these distributions we find that for any $\phi$ that since
\begin{align*}
X_{\theta+\phi} & =\cos(\theta+\phi)X+\sin(\theta+\phi)Y \\ & = \cos(\theta)\cos(\phi)X-\sin(\theta)\sin(\phi)X + \cos(\theta)\sin(\phi)Y+\sin(\theta)\cos(\phi)Y \\ & = \cos(\theta)X_{\phi}+\sin(\theta)X_{\phi+\pi/2},
\end{align*}
we have
$$
p = \pr(f(X_\phi)\neq f(X_{\phi+\theta})).
$$
Therefore we have for any integer $n$ that
\begin{equation}\label{manychanges}
np = \pr(f(X_0)\neq f(X_\theta))+\pr(f(X_\theta)\neq f(X_{2\theta}))+\ldots+\pr(f(X_{(n-1)\theta})\neq f(X_{n\theta})).
\end{equation}
We define the random function $F:\R\rightarrow \{-1,1\}$ by
$$
F(\phi)=f(X_\phi).
$$
($F$ depends on $X$ and $Y$ as well as $\phi$).
We note that the left hand side of Equation \ref{manychanges} is at most the number of times that $F(\phi)$ changes signs on the interval $[0,n\theta]$.  Therefore we have that
\begin{equation}\label{signchange}
p \leq \frac{\E[\textrm{number of times} \ F \ \textrm{changes signs on} \ [0,n\theta]]}{n}.
\end{equation}
We note that $F(\phi)$ is periodic in $\phi$ with period $2\pi$.  Therefore the number of times $F$ changes sign on $[0,n\theta]$ is the number of times that $F$ changes sign on $[0,2\pi)$ times $\left(\frac{n\theta}{2\pi}+O(1)\right)$.  Applying this to Equation \ref{signchange}, we get that
$$
p \leq \frac{\E[\textrm{number of sign changes of} \ F \ \textrm{on} \ [0,2\pi)]\left(\frac{n\theta}{2\pi} +O(1) \right)}{n}.
$$
Taking a limit as $n\rightarrow\infty$ yields
\begin{equation}\label{finalbound}
p\leq \frac{\theta\E[\textrm{number of sign changes of} \ F \ \textrm{on} \ [0,2\pi]]}{2\pi}.
\end{equation}

We now make use of the fact that $f$ is a degree $d$ polynomial threshold function.  In particular we will show that for any $X$ and $Y$ that $F$ changes signs at most $2d$ times on $[0,2\pi)$.  We let $f=\sgn(g)$ for some degree $d$ polynomial $g$.  We note that the number of sign changes of $F$ is equal to the number of zeroes of the function $g(\cos(\phi)X+\sin(\phi)Y)$ (unless this function is identically 0, which happens with probability 0 and can be ignored).  It should be noted though that $g(\cos(\phi)X+\sin(\phi)Y)=0$ if and only if $z=e^{i\phi}$ is a root of the degree-$2d$ polynomial
$$
z^d g\left(\left(\frac{z+z^{-1}}{2}\right)X+\left(\frac{z-z^{-1}}{2i}\right)Y \right).
$$
Therefore the expectation in Equation \ref{finalbound} is at most $2d$.  Therefore we have
$$
p \leq \frac{2d\theta}{2\pi} = \frac{d\theta}{\pi}
$$
as desired.

We also note the ways in which the above bound can fail to be tight.  Firstly, there may be some probability that $F$ changes signs less than $2d$ times on a full circle.  Secondly, the number of times $F$ changes signs may be more than the fraction of the time that $f(n\theta)\neq f((n+1)\theta)$ if sign changes are spaced more tightly than $\theta$.  On the other hand it should be noted that if $f$ is the threshold function for a product of $d$ distinct homogeneous linear functions, the first case happens with probability 0, and the probability of the second case occurring will necessarily go to 0 as $\epsilon$ does.  Therefore for such functions our bound is asymptotically correct as $\epsilon\rightarrow 0$.
\end{proof}

\section{Proof of the Gaussian Surface Area Bounds}

We will first need to bound a slight variant of the noise sensitivity of a polynomial threshold function.  We begin by proving the following Lemma:

\begin{lem}\label{linelem}
If $f$ is a degree $d$ polynomial threshold function in $n$ dimensions, $\epsilon>0$ and $X$ a random Gaussian variable, then
$$
\pr(f(X)\neq f(X(1+\epsilon))) \leq d\epsilon\sqrt{\frac{n}{4\pi}}.
$$
\end{lem}
\begin{proof}
First note that by first conditioning on the line that $X$ lies in we may reduce this problem to the case of a one dimensional distribution.  Note that $f$ changes sign at most $d$ times along this line.  We need to bound the probability that at least one of these sign changes is between $X$ and $(1+\epsilon)X$.  It therefore suffices to prove that for any one of these sign changes, that it lies between $X$ and $(1+\epsilon)X$ with probability at most $\epsilon\sqrt{\frac{n}{4\pi}}$.  Note that the probability that $X$ is on the correct side of the origin is $\frac{1}{2}$.  Beyond that $|X|^2$ satisfies the $\chi^2$ distribution with $n$ degrees of freedom, namely $\frac{1}{2^{n/2}\Gamma(n/2)}x^{n/2-1}e^{-x/2}dx$.  Letting $y=\log(x)=2\log(|X|)$ we find that that $y$ has distribution
$$
\frac{1}{2^{n/2}\Gamma(n/2)}e^{ny/2}e^{-e^y / 2}dy.
$$
We want the probability that $y$ is within a particular range of size $2\log(1+\epsilon)$.  This is at most $2\epsilon$ times the densest part of the density function.  This is achieved when $ny - e^y$ is maximal, or when $y=\log(n)$.  Then the density is
$$
\frac{1}{2^{n/2}\Gamma(n/2)}n^{n/2}e^{-n/2} = \sqrt{\frac{n}{4\pi}}\frac{(n/2)^{n/2}e^{-n/2}\sqrt{2\pi (n/2)}}{(n/2)!} \leq \sqrt{\frac{n}{4\pi}}.
$$
Multiplying this by $d$, $2\epsilon$ and $\frac{1}{2}$ (the probability that $X$ is on the correct side of 0), we get our bound.  Notice also that this bound should be nearly sharp if the polynomial giving $f$ is a product of terms of the form $|X|^2-r_i$ for $r_i$ approximately $n$ and spaced apart by factors of $(1+\epsilon)^2$.
\end{proof}

We can now prove a bound on a quantity more relevant to Gaussian surface area:

\begin{cor}\label{wigglecor}
If $f$ is an $n$ dimensional, degree $d$ polynomial threshold function, $\epsilon>0$ and $X$ and $Y$ independent Gaussians, then
$$
\pr(f(X)\neq f(X+\epsilon Y)) \leq \frac{d\epsilon}{\pi} + \frac{d\epsilon^2}{4}\sqrt{\frac{n}{\pi}}.
$$
\end{cor}
\begin{proof}
We let $r=\sqrt{1+\epsilon^2}$, $\theta=\arctan(\epsilon)$, and let $Z=\cos(\theta)X+\sin(\theta)Y$ be a normal random variable.  Note that $X+\epsilon Y= rZ$.  We then have that
$$
\pr(f(X)\neq f(X+\epsilon Y)) \leq \pr(f(X) \neq f(Z)) + \pr(f(Z) \neq f(rZ)).
$$
By Theorem \ref{nsthm} and Lemma \ref{linelem} this is at most
$$
\frac{d\theta}{\pi} + d(r-1)\sqrt{\frac{n}{4\pi}} \leq \frac{d\epsilon}{\pi} + \frac{d\epsilon^2}{4}\sqrt{\frac{n}{\pi}}.
$$
\end{proof}

In particular, we relate this to Gaussian surface area by:

\begin{lem}\label{crossarealem}
If $f$ is a boolean function with $f^{-1}(1)$ open with smooth boundary and Gaussian area $S$, and if $X$ and $Y$ are independent Gaussians then,
\begin{equation}\label{crossbound}
\lim_{\epsilon \rightarrow 0} \frac{\pr(f(X)=-1 \ \textrm{and} \ f(X+\epsilon Y)=1)}{\epsilon} = \frac{S}{\sqrt{2\pi}}.
\end{equation}
\end{lem}
\begin{proof}
First note that if $A=f^{-1}(1)$, then
$$
\lim_{\epsilon\rightarrow 0} \frac{\textrm{GaussianVolume}(A_\epsilon\backslash A)}{\epsilon}= S
$$
rather than just the liminf being equal.  Next note that since the probability that $|\epsilon Y|>\epsilon^{2/3}$ goes rapidly to 0 as $\epsilon\rightarrow 0$, we can throw away all cases where $X$ is not within $\epsilon^{2/3}$ of $\partial A$ from the left hand side of Equation \ref{crossbound}.  When $X$ is close to $\partial A$ and when $f(X)=-1$, we may approximate the probability that $f(X+\epsilon Y)=1$ by the probability that the component of $\epsilon Y$ in the direction of the shortest path from $X$ to $A$ is more than $d(X,A)$.  Since $\partial A$ is smooth, this approximation is accurate for $X$ close to $A$, and in particular for $X$ within $\epsilon^{2/3}$ should introduce an error of $O(\epsilon^{4/3})$, which can be ignored.  Hence if $Z$ is a normalized one variable Gaussian, the numerator of left hand side can be replaced by
$$
\pr(\epsilon Z \geq d(X,A)>0) = \pr\left(Z\geq \frac{d(X,A)}{\epsilon}>0\right).
$$
This is easily seen to be
\begin{align*}
\int_0^\infty \frac{1}{\sqrt{2\pi}}e^{-x^2/2}\pr(0<d(X,A)\leq \epsilon x) dx & = \int_0^\infty \frac{1}{\sqrt{2\pi}}e^{-x^2/2} \textrm{GVol}(A_{\epsilon x}\backslash A)dx\\ & = \int_0^\infty \frac{1}{\sqrt{2\pi}}e^{-x^2/2}S\epsilon x (1+o(1))dx\\ & = \frac{S\epsilon}{\sqrt{2\pi}} + o(\epsilon).
\end{align*}
Thus completing our proof.
\end{proof}

\begin{proof}[Proof of Theorem \ref{surfaceareatheorem}]
This follows immediately from Corollary \ref{wigglecor} and Lemma \ref{crossarealem} after noting that
$$\pr(f(X)=-1, f(X+\epsilon Y)=1) \sim \frac{1}{2}\pr(f(X)\neq f(X+\epsilon Y))\sim \frac{d\epsilon}{2\pi}.$$
\end{proof}

\section{Conclusion}

We have shown nearly tight bounds on the Gaussian surface area and noise sensitivity of polynomial threshold functions.  One might hope to generalize these results to work for other distributions, such as the uniform distribution on vertices of the hypercube.  Unfortunately, several aspects of this proof are difficult to generalize.  Perhaps most significantly, we lose the symmetry that allowed us to prove our original result on noise sensitivity.  Another difficulty would be in the relation between noise sensitivity and surface area.  In our case the two are essentially equivalent quantities of study.  On the other hand \cite{kn:learning} defined a notion of surface area for the hypercube distribution and proved that for even linear threshold functions there could be a gap between noise sensitivity and surface area of as much as $\Theta(\sqrt{\log(n)})$.

\end{document}